\documentclass[a4paper, 12pt]{elsarticle}

\usepackage[english]{babel}
\usepackage{babel}
\usepackage{fontenc}
\usepackage{graphicx}
\usepackage{amsmath,amsthm,amssymb}
\usepackage{natbib}
\usepackage[colorlinks=true,citecolor=blue]{hyperref}
\usepackage{color}
\usepackage{graphicx}
\usepackage{subfigure}
\usepackage{color}
\usepackage{newlfont}
\usepackage{multirow}
\usepackage{longtable} 
\usepackage{ifthen}
\usepackage{alltt}
\usepackage{enumerate}
 \usepackage[text={6.25in,8.5in},centering]{geometry}
\usepackage{tikz}
\usepackage{amssymb}
\usepackage{amsthm}
\usepackage{amsmath}
\usepackage{tikz}
\usetikzlibrary{shapes,arrows}
\numberwithin{equation}{section}
\newtheorem{theorem}{Theorem}[section]

\input epsf.sty
\usepackage{lineno}

\begin{document}

\title{Within host dynamics of SARS-CoV-2 in humans: Modeling immune responses and antiviral treatments}
\author[ISI]{Indrajit Ghosh\footnote{Corresponding author. Email: indra7math@gmail.com, indrajitg\_r@isical.ac.in}}
\address[ISI]{Agricultural and Ecological Research Unit, Indian Statistical Institute, Kolkata, West Bengal 700108, India}

\begin{abstract}
In December 2019, a newly discovered SARS-CoV-2 virus was emerged from China and propagated worldwide as a pandemic, resulting in about 3–5\% mortality. In the absence of preventive medicine or a ready to use vaccine, mathematical models can provide useful scientific insights about transmission patterns and targets for drug development. In this study, we propose a within-host mathematical model of SARS-CoV-2 infection considering innate and adaptive immune responses. We analyze the equilibrium points of the proposed model and obtain an expression of the basic reproduction number. We then numerically show the existence of a transcritical bifurcation. The proposed model is calibrated to real viral load data of two COVID-19 patients. Using the estimated parameters, we perform global sensitivity analysis with respect to the peak of viral load. Finally, we study the efficacy of antiviral drugs and vaccination on the dynamics of SARS-CoV-2 infection. Our results suggest that blocking the production of the virus by infected cells decreases the viral load more than reducing the infection rate of healthy cells. Vaccination is also found useful but during the vaccine development phase, blocking virus production from infected cells can be targeted for antiviral drug development.  
\end{abstract}


\maketitle

\section{Introduction}\label{sec1}
Coronaviruses are a large group of viruses that have the potential to transmit between hosts. These are enveloped in positive-sense, non-segmented RNA viruses belonging to the Coronaviridae family (Nidovirales order) and widely distributed in humans and other mammals \cite{huang2020clinical}. The virus is responsible for a range of symptoms including fever, cough, and shortness of breath \cite{huang2020clinical}. Some patients have reported radiographic changes in their ground-glass lungs, healthy or lower than average white blood cell lymphocyte, and platelet counts; hypoxaemia; and deranged liver and renal function. Since first discovery and identification of coronavirus in 1965, three significant outbreaks occurred, caused by emerging, highly pathogenic coronaviruses, namely the 2003 outbreak of "Severe Acute Respiratory Syndrome" (SARS) in mainland China \cite{gumel2004modelling,li2003angiotensin}, the 2012 outbreak of "Middle East Respiratory Syndrome" (MERS) in Saudi Arabia \cite{de2013commentary,de2016sars}, and the 2015 outbreak of MERS in South Korea \cite{cowling2015preliminary,kim2017middle,sardar2020realistic}. These outbreaks resulted in SARS and MERS cases confirmed by more than 8000 and 2400, respectively \cite{kwok2019epidemic}. A newer and genetically similar coronavirus is responsible for the coronavirus disease 2019 (COVID-19). The virus is named SARS-CoV-2. Despite a relatively lower case fatality rate compared to SARS and MERS, the COVID-19 spreads rapidly and infects more people than the SARS and MERS. Despite strict intervention measures implemented in the region where the COVID-19 was originated, the infection spread locally and elsewhere very rapidly. COVID-19 has been declared a pandemic by the World Health Organization in January 2020. Since its first isolation in Wuhan, China in December 2019, it has caused outbreak with more than 10 million confirmed infections and above 500 thousand reported deaths worldwide as of 28 June 2020. The affected countries around the globe are fighting the virus by implementing social distancing and isolation strategies. Unfortunately, the COVID-19 has neither a preventive medicine nor a ready to use vaccine. Multiple approaches are adopted in the development of Coronavirus vaccines; most of these targets the surface-exposed spike (S) glycoprotein or S protein as the primary inducer of neutralizing antibodies \cite{dhama2020covid,lurie2020developing}. In fact, either monoclonal antibody or vaccine approaches have failed to neutralize and protect from previous coronavirus infections \cite{menachery2015sars}. Therefore, individual behaviour (e.g. early self-isolation and social distancing), as well as preventive measures such as hand washing, covering when coughing, are critical to control the spread of COVID-19 \cite{vargas2020host}. However, researchers have been putting more effort into finding a solution to this pandemic situation \cite{hu2020insights,yaqinuddin2020innate,tay2020trinity}.

In addition to medical and biological research, theoretical studies based on mathematical models may also play an important role throughout this anti-epidemic fight in understanding the epidemic character traits of the outbreak, in having to decide on the measures to reduce the spread and in understanding within-host patterns of virus transmission. While there are many mathematical models developed at an epidemiological level for COVID-19 \cite{wu2020nowcasting,tang2020updated,kucharski2020early,kochanczyk2020dynamics}, there are very few within-host level studies to understand SARS-CoV-2 replication cycle and its interactions with the innate and adaptive immune responses \cite{du2020mathematical,vargas2020host}. In these few previous studies, authors studied target cell models and target cell models with eclipse phase. Therefore, detailed research with immune responses is necessary for the understanding of SARS-CoV-2 spread inside the human body. The human immune system is comprised of innate and adaptive immune responses. While the adaptive immune system is both fast and effective at targeting invasions by previously encountered pathogens, its role in host defence in the first days of a new infection is secondary to that of the innate immune system.

Motivated by this discussion, we aim to develop a within-host mathematical model of SARS-CoV-2 infection considering human immune responses. This model can be used as a basis for understanding characterized patterns of disease severity in humans. Moreover, we intend to use real viral load data from COVID-19 positive patients to calibrate the proposed model so that the parameters are realistic for further inference. The main goal is to compare the efficacy of various antiviral drugs and identify the most beneficial target. 

The rest of the paper is organized as follows: in Section \ref{model_description}, we formulate the compartmental model of within human SARS-CoV-2 transmission; the equilibrium points of the proposed model are analyzed and the basic reproduction number is obtained in Section \ref{equilibria}; viral load time series, transcritical bifurcation, fitting model to real data and global sensitivity analysis are presented in Section \ref{numerical}; in Section \ref{treatment}, we study the efficacy of antiviral drugs and vaccination; finally, the obtained results are discussed in Section \ref{discussion}. 

\section{The mathematical model}\label{model_description}
A deterministic ordinary differential equation model describing cell--virus--immune response interaction dynamics of SARS-CoV-2 infection is being formulated. Time-dependent state variables are taken to represent the compartments. A general mathematical model for the underlying dynamics of virus-host cell interaction has been studied in this context \cite{du2020mathematical,vargas2020host}. However, the basic principles that underlay models of virus dynamics are as follows: Healthy uninfected cells, $H(t)$, are infected when they meet free viruses, $V(t)$. Infected cells, $I(t)$, produce new virus particles that leave the cell and find other susceptible target cells. Whenever a human is infected with SARS-CoV-2, his innate and adaptive immune responses work together to neutralize the threat of SARS-CoV-2 infection \cite{tufan2020covid,mckechnie2020innate,tay2020trinity}. The innate immune response works non-specifically and immediately after the viral attack. Cells and proteins of the innate immune system are ever-present in a healthy host and can respond to invading pathogens within the first minutes and hours of infection \cite{ciupe2017host}. This system is of great importance in the sense that it is preventing the establishment of new infections during the activation time of the adaptive immune system. It is believed that Cytokines are an essential component of the immune system \cite{sasmal2017mathematical}. They are a family of small soluble proteins secreted by different cells. They can be loosely classified into one of four families: the haematopoietins, the immunoglobin superfamily, the tumour necrosis factor family and the interferons (IFN). Cytokines modulate the balance between innate and adaptive immune responses. The IFNs are perhaps the most critical cytokines in the initial innate response to viral infection. They are classified into two types: IFN-$\alpha$ (a family of related proteins) and the single protein IFN-$\beta$ together form type I; IFN-$\gamma$ is the sole and unrelated type II IFN. IFN-$\alpha$ and IFN-$\beta$ are secreted by cells in response to viral infection and promote an antiviral response in otherwise susceptible cells. Cytokines C(t) is vital in inhibiting viral replication and modulating downstream effects of the immune response. Specific cytokines activate natural killer (NK) cells N(t), which play an essential role in killing virus-infected cells. As in \cite{canini2011population,ben2015minimal}, the rate of NK cells increment by cytokines is taken as $r C$, whereas NK cells die at a rate $\mu_5$. However, Against the inhibiting mechanism of cytokines, the viruses often target the JaK/STAT pathway to decrease the production of IFNs. This mechanism, known as immunosupression, is observed for SARS-CoV-2 \cite{raoult2020coronavirus}. The functional form of a decrease in the cytokine production rate is assumed to be $\frac{k_2 I}{1 + \gamma V}$. 

Meanwhile, cytokines also activate the adaptive immune system, mainly the cytotoxic T-lymphocytes T(t) at a rate $\lambda_1$. Interleukin-2 (IL-2) is a type of cytokine signaling molecule in the immune system that is very important to activate T-cells. T-cells finds virus infected cells and kill them at a rate $p_1$. T-cells subsequently activate B-lymphocytes B(t) at a rate $\lambda_2$ to produce antibody against the virus. B-cells mainly secrete IgM and IgG antibodies that are released in the blood and lymph fluid, where they specifically recognize and neutralize the SARS-CoV-2 viral particles \cite{sasmal2017mathematical,tufan2020covid}. Meanwhile, antibody levels A(t) are increasing with the aim of halting infection (and in future providing protection against a subsequent infection). A schematic flow diagram of the model is depicted in Fig. \ref{Fig:scematic_flow_diagram}.

\begin{figure}[ht]
	\includegraphics[width=1.0\textwidth]{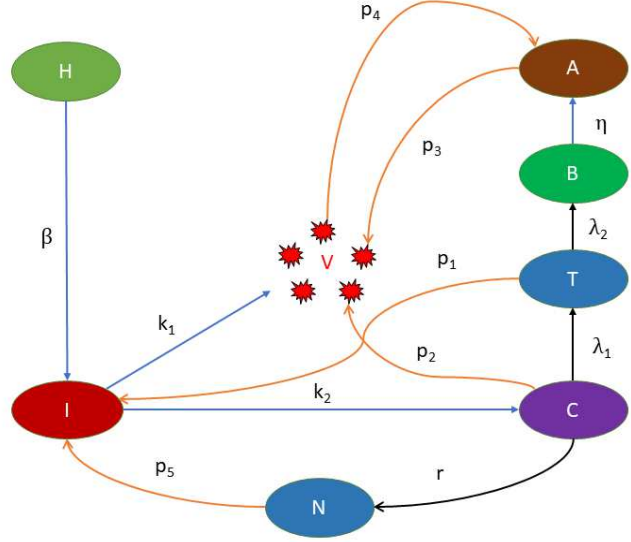}
	\caption{Schematic diagram of the proposed model. The blue arrows indicate production, black arrows indicate activation and orange ones show inhibition by different cells.}
	\label{Fig:scematic_flow_diagram}
\end{figure}

Finally, the cell--virus--immune response interaction dynamics of SARS-CoV-2 infection are governed by the following system of differential equations:

\begin{eqnarray}\label{eq1}
\frac{dH}{dt}&=& \Pi-\beta H V -\mu_1 H,\nonumber\\
\frac{dI}{dt}&=& \beta H V- p_1 T I - p_5 N I - \mu_2 I,\nonumber\\
\frac{dV}{dt}&=& k_1 I - p_2 CV - p_3 AV - \mu_3 V,\nonumber\\
\frac{dC}{dt}&=& \frac{k_2 I}{1 + \gamma V} - \mu_4 C,\\
\frac{dN}{dt}&=& rC  - \mu_5 N,\nonumber\\
\frac{dT}{dt}&=& \lambda_1 C T - \mu_6 T,\nonumber\\
\frac{dB}{dt}&=& \lambda_2 T B - \mu_7 B,\nonumber\\
\frac{dA}{dt}&=& G(t-\tau)\eta B - p_4 AV -\mu_8 A.\nonumber
\end{eqnarray}

The time delay $\tau$ introduced through the Heaviside step function \cite{fowler1981approximate}, is the time period that is required for the first production of antibodies after the T-lymphocytes and B-lymphocytes interact. This delay is biologically significant since the production of antibodies after the virions have associated with the B-lymphocytes is a complex process involving multiple steps. The B-cells have to undergo differentiations before they can be transformed into the plasma cells capable of producing antibodies \cite{gujarati2014virus}. The Heaviside step function $G(t)$ is defined as follows,

\begin{align*}
    G(t-\tau) &= 1, \text{if} \; \; t>\tau\\
              &= 0, \text{if} \; \; t<\tau
\end{align*}

The model \ref{eq1} has initial conditions given by:
$H(0)=H_0 \geq 0$, $I(0)=I_0 \geq 0$, $V(0)=V_0 \geq 0$, $C(0)=C_0 \geq 0$, $N(0)=N_0 \geq 0$, $T(0)=T_0 \geq 0$, $B(0)=B_0 \geq 0$, and $A(0)=A_0 \geq 0$.\\

\begin{table}
\caption{Parameters used in model \ref{eq1} }
\label{table2}
\centering
\begin{tabular}{|p{6cm}|c|p{6cm}|c|}
\hline
Parameter & Symbol & value/Range & Reference \\
\hline
Production rate of healthy cells & $\Pi$ & 4 $\times$ $10^3$ cells ml$^{-1}$ day$^{-1}$ & \cite{nikin2015role}\\
Rate at which healthy cells are converted to infected cells & $\beta$ & (5 -- 561) $\times$ $10^{-9}$ ml (RNA copies)$^{-1}$ day$^{-1}$ & \cite{vargas2020host} \\
Strength of immunosupresion & $\gamma$ & 0.5 ml (RNA copies)$^{-1}$ & Assumed\\
Rate at which T-cells destroy infected cells & $p_1$ & 0.001 ml cells$^{-1}$ day$^{-1}$ & \cite{clapham2014within} \\
Rate at which viral particles are neutralized by cytokines & $p_2$ & (0 -- 1) ml cells$^{-1}$ day$^{-1}$ & Estimated\\
Rate at which viral particles are neutralized by antibodies & $p_3$ & (0 -- 1) ml molecules$^{-1}$ day$^{-1}$ & Estimated \\
Rate at which virus neutralize antibodies & $p_4$ & 3 $\times$ $10^{-7}$ ml (RNA copies)$^{-1}$ day$^{-1}$ & \cite{nikin2015role} \\
Rate at which infected cells are diminished by NK cells & $p_5$ & 5.74 $\times$ $10^{-4}$ ml cells$^{-1}$ day$^{-1}$  & \cite{ben2015minimal} \\
Production rate of virus from infected cells & $k_1$ & (8.2 -- 525) day$^{-1}$ & \cite{vargas2020host} \\
Production rate of cytokines & $k_2$ & (0 -- 10) day$^{-1}$ & Assumed \\
Activation rate of NK cells & $r$ & 0.52 day$^{-1}$ & \cite{ben2015minimal} \\
Activation rate of T cells & $\lambda_1$ & 0.1 ml cells$^{-1}$ day$^{-1}$ & \cite{sasmal2017mathematical}\\
Activation rate of B cells & $\lambda_2$ & 0.01 ml cells$^{-1}$ day$^{-1}$ & \cite{sasmal2017mathematical}\\
Rate at which antibodies are produced & $\eta$ & (0 - 1) day$^{-1}$ & \cite{gujarati2014virus} \\
Natural death rate of Healthy cells and protected cells & $\mu_1$ & 0.14 day$^{-1}$ & \cite{sasmal2017mathematical} \\
Natural death rate of infected cells & $\mu_2$ & (0 -- 1) day$^{-1}$ & Assumed \\
Clearance rate of virus & $\mu_3$ & (0 -- 1) day$^{-1}$ & Estimated \\
Natural death rate of cytokines & $\mu_4$ & 0.7 day$^{-1}$ & Assumed \\
Natural death rate of NK cells & $\mu_5$ & 0.07 day$^{-1}$ & \cite{ben2015minimal} \\
Natural death rate of T cells & $\mu_6$ & 1 day$^{-1}$ & \cite{sasmal2017mathematical} \\
Natural death rate of B cells & $\mu_7$ & 0.2 day$^{-1}$ & \cite{nikin2015role} \\
Natural death rate of antibodies & $\mu_8$ & 0.07 day$^{-1}$ & \cite{sasmal2017mathematical} \\
Time delay for antibody production & $\tau$ & 7 -- 14 days & \cite{immunity2020who} \\
\hline
\end{tabular}
\end{table}

\section{Equilibria and Basic reproduction number}\label{equilibria}
There are four type of equilibia of the system \eqref{eq1}, namely,

(a) The disease free equilibrium (DFE) given by $E_0=(\frac{\Pi}{\mu_1},0,0,0,0,0,0,0)$.

(b) The virus persistence equilibrium in the absence of immune responses, given by
$E_1=(H_1,I_1,V_1,0,0,0,0,0)$, where 
$H_1=\frac{\Pi}{\mu_1 R_0}$, $I_1=\frac{\mu_1 \mu_3}{\beta k_1} (R_0 - 1)$ and $V_1=\frac{\mu_1}{\beta} (R_0 - 1)$ with $R_0=\displaystyle \frac{\Pi \beta k_1}{\mu_1 \mu_2 \mu_3}$. Clearly, this equilibrium exists only when $R_0>1$.

(c) The virus persistence equilibrium in the absence of adaptive immune responses, given by $E_2=(H_2,I_2,V_2,C_2,N_2,0,0,0)$, where (assume, $Q=\beta H_2 V_2$) 
$H_2=\frac{\Pi - Q}{\mu_1}$, $N_2=\frac{r C}{\mu_5}$, $I_2=\frac{Q}{ \mu_2 + p_5 N_2}$, $V_2= \frac{1}{\gamma} \left[\frac{k_2 I_2}{\mu_4 C_2}-1\right]$ and $C_2$ is given by the roots of the following cubic equation
$$\frac{p_2 p_5 \mu_4 r}{\mu_5} C_2^3 + (\mu_2 \mu_4 p_2 + \frac{\mu_3 \mu_4 p_5 r}{\mu_5} )C_2^2+(\mu_2 \mu_3 \mu_4 + \mu_4 \gamma k_1 Q - k_2 p_2 Q)C_2 - k_2 \mu_3 Q=0$$. 

Note that, irrespective of the sign of the coefficient of $C_2$, Descartes' rule of sign ensure existence of exactly one positive root whenever $\frac{k_2 I_2}{\mu_4 C_2}>1$.

(d) The all cells and virus co-existence equilibrium, given by
$E_3=(H_3,I_3,V_3,c_3,N_3,T_3,B_3,A_3)$, where (assume, $Q=\beta H_3 V_3$) 
$H_3=\frac{\Pi - Q}{\mu_1}$, $I_3=\frac{\mu_4 \mu_6 R_1}{\lambda_1 k_2}$, $V_3= \frac{1}{\gamma}\left[ R_1 - 1\right]$, $C_3=\frac{\mu_6}{\lambda_1}$, $N_3=\frac{r C_3}{\mu_5}$, $T_3=\frac{\mu_7}{\lambda_2}$, $B_3=\frac{A_3}{\eta} [p_4 V_3 + \mu_8]$ and 
$A_3= \frac{1}{p_3 V_3} [R_2 - 1]$, with

$$R_1=\frac{\lambda_1 k_2 Q}{\mu_4 \mu_6 (\frac{p_1 \mu_7}{\lambda_2} + \frac{r p_5 \mu_6}{\lambda_1 \mu_5}+ \mu_5)} $$ and 
$$R_2=\frac{\gamma \lambda_1^2 k_1 k_2 Q}{R_1 \mu_4 \mu_6(\lambda_1 \mu_3 + p_2 \mu_6)(R_1-1)}.$$ 

It can be noted that this equilibrium exists only when $R_1>1$ and $R_2>1$.

\begin{theorem}\label{theorem_R1}
The DFE $E_0$ of the system (\ref{eq1}) is locally asymptotically stable, if $R_0<1$, and unstable if $R_0>1$, where
\begin{equation}
\centering
\displaystyle R_0=\displaystyle \frac{\Pi \beta k_1}{\mu_1 \mu_2 \mu_3}.
\label{EQ:eqn_R0}
\end{equation}
\end{theorem}

\begin{proof}
The Jacobian of the system (\ref{eq1}) at $E_0$ is given as \\
\begin{eqnarray}
\begin{array}{llllllll}
J(E_0)=\begin{bmatrix}
    -\mu_1 & 0 & -\frac{\beta \Pi}{\mu_1} & 0 & 0 & 0 & 0 & 0\\
    0 & -\mu_2 &  \frac{\beta \Pi}{\mu_1} & 0 & 0 & 0 & 0 & 0\\
    0 & k_1 & -\mu_3 & 0 & 0 & 0 & 0 & 0\\
    0 & k_2 & 0 & -\mu_4 & 0 & 0 & 0 & 0\\
    0 & 0 & 0 & r & -\mu_5  & 0 & 0 & 0\\
    0 & 0 & 0 & 0 & 0 &  -\mu_6  & 0 & 0\\
    0 & 0 & 0 & 0 & 0 & 0 & -\mu_7 & 0\\
    0 & 0 & 0 & 0 & 0 & 0 & G(t-\tau)\eta & -\mu_8\\
\end{bmatrix}\\\\
\end{array}
\label{EQ:eqn 1.6}
\end{eqnarray}

Clearly, $-\mu_1$, $-\mu_4$, $-\mu_5$, $-\mu_6$, $-\mu_7$ and $-\mu_8$ are eigenvalues of this Jacobean matrix and other two eigenvalues are given by the roots of the following equation\\

\begin{eqnarray}
C(\Lambda):= \Lambda^2 + a_1 \Lambda +a_2 &=& 0\nonumber\\
\end{eqnarray}

where
\begin{align}
a_1 &= \mu_2 + \mu_3\nonumber\\
a_2 &= \mu_2 \mu_3 \left(1-R_0\right)\nonumber\\
\end{align}

Therefore, for $R_0<1$, the conditions for the Routh-Hurwitz criteria are satisfied and hence DFE is locally asymptotically stable. Now if $R_0>1$, then $a_2<0$ and $C(\lambda)=0$ will possess a positive real solution. Therefore the DFE will be unstable for $R_0>1$. Hence the proof follows.
\end{proof}

The stability of the other three equilibrium points is complicated and does not lead to biologically relevant stability conditions. Therefore, we explore model solutions, relevant model dynamics, important parameters, agreement with real data through numerical simulations.

\section{Numerical Simulation}\label{numerical}
In this section, important properties of the proposed model are investigated numerically. Using different parameter settings, time series and threshold analysis is performed. Moreover, the agreement of the model solution with real data is explored. Through out this section the following set of initial conditions is used unless stated $H(0)=4 \times 10^5$ cells per ml, $I(0)=3 \times 10^{-4}$ cells per ml, $V(0)=357$ RNA copies per ml, $C=0$ cells per ml, $N=100$ cells per ml, $T=500$ cells per ml, $B=100$ cells per ml and $A=0$ molecules per ml (most of the initial conditions are taken from \cite{sasmal2017mathematical,nikin2015role}).

\subsection{Time series and threshold analysis}
We first study the time series of the viral load and antibody count. In Fig. \ref{Fig:time_series_VA}, the viral load and antibody are plotted. The viral load time series experiences a peak between sixth and seven days post infection. However, as soon as the adaptive immune response is activated (after $\tau=7$ days), a sharp decrease is observed in the viral load. On the other hand, the antibody count starts to rise after 7 days post infection and shows saturated type behaviour.

\begin{figure}[t]
	\includegraphics[width=0.45\textwidth]{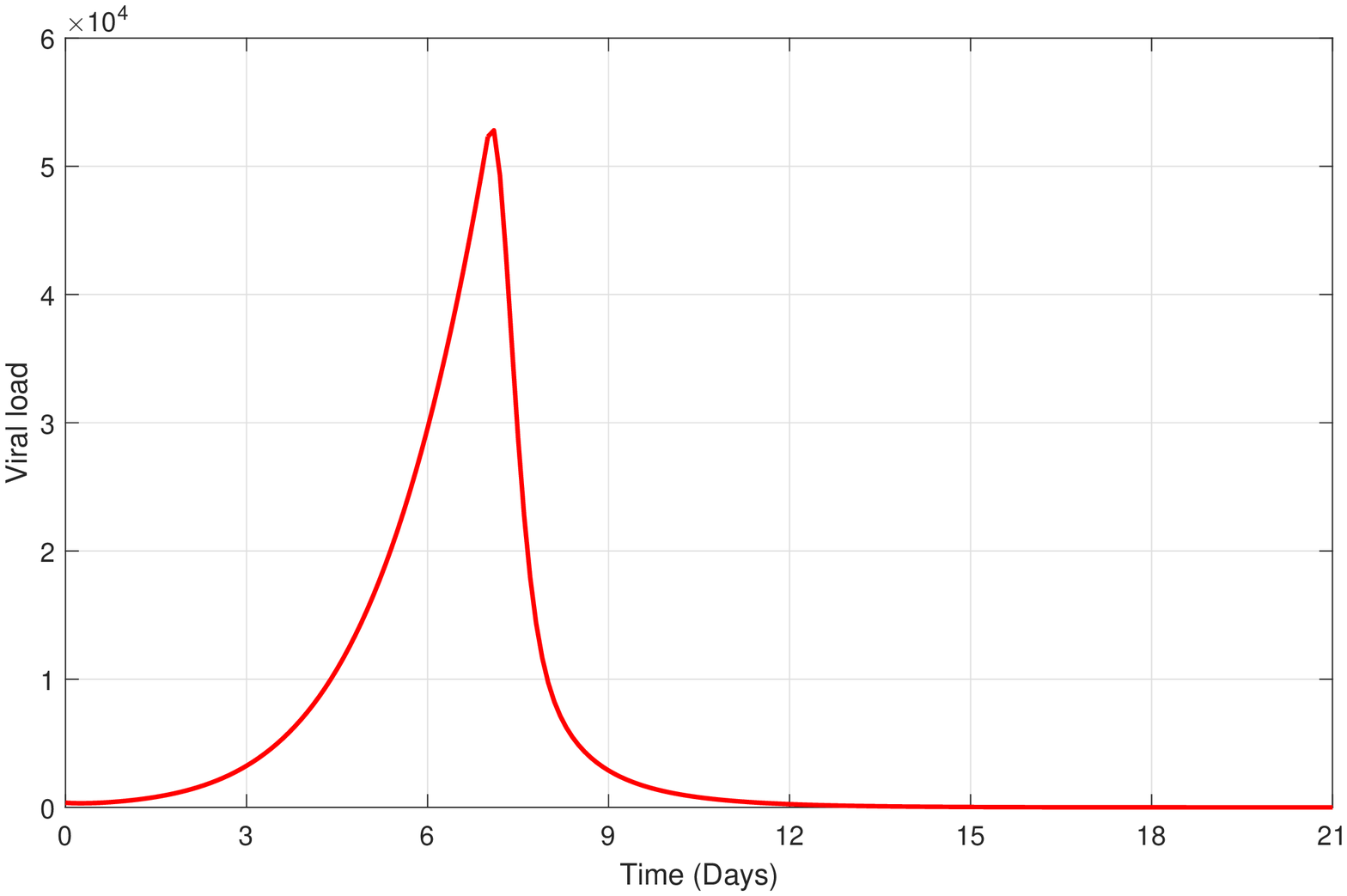}{a}
	\includegraphics[width=0.45\textwidth]{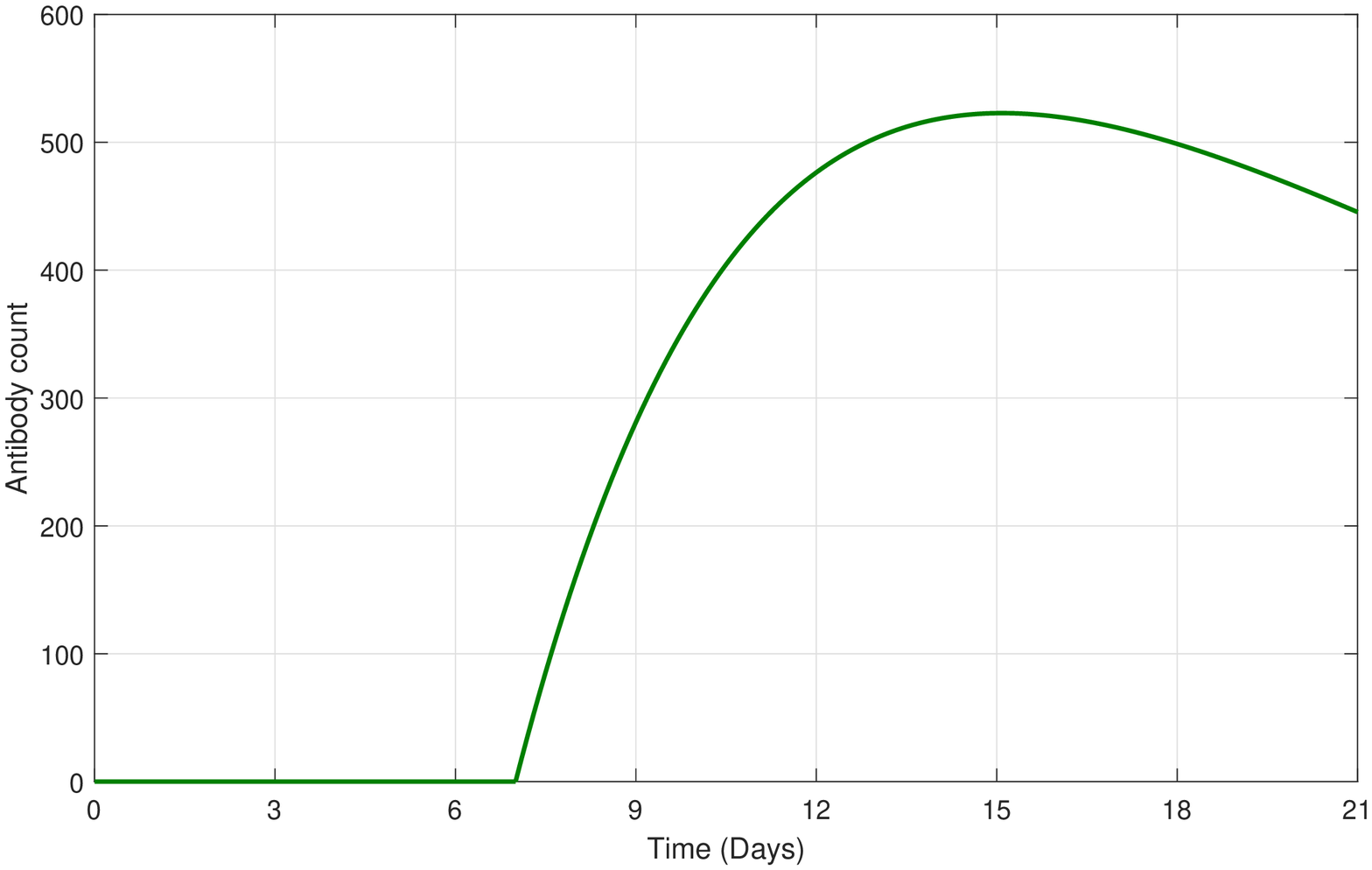}{b}
	\caption{Time evolution of (a) viral load ($V$) and (b) antibody count ($A$) of the model \ref{eq1}. All the parameters are taken from Table \ref{table2} except $\beta = 2 \times 10^{-8}$, $\mu_2 =0.65$, $\mu_3=0.9$, $p_2=0.001$, $p_3=0.05$, $k_1=500$, $k_2=5$, $\eta=0.05$ and $\tau=7$.}
	\label{Fig:time_series_VA}
\end{figure}

Further, we study the threshold for $R_0$. It is observed that $R_0=1$ acts as a critical value for the persistence of virus particles. The virus particles converges to the DFE of the model \ref{eq1} for $R_0<1$ and the viral load converges to a non-zero value as soon as $R_0$ crosses unity. This type of phenomenon is called forward bifurcation where the two equilibrium points switches their stability at a critical value. The diagram is depicted in Fig. \ref{Fig:threshold_r0}. This also ensure that if we vary other parameters involved in the expression of $R_0$, the same type of phenomenon occurs. Thus, in turn parameters such as $\beta$ and $k_1$ can be reduced so as to reduce $R_0$ below unity. 

\begin{figure}[t]
	\includegraphics[width=1.0\textwidth]{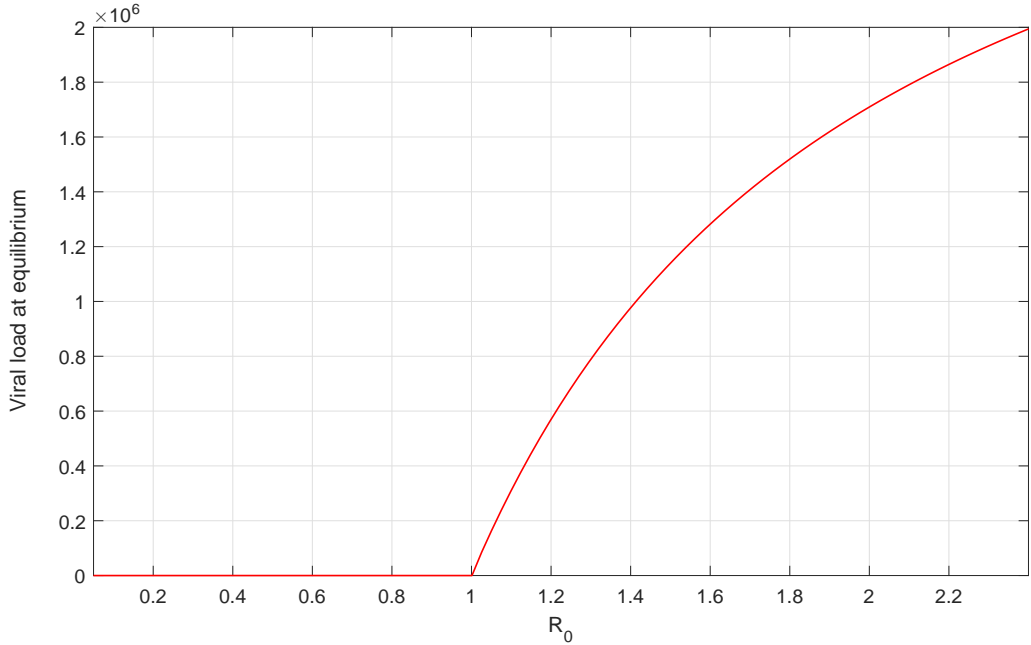}
	\caption{Forward bifurcation diagram with respect to basic reproduction number. All the fixed parameters are taken from Table \ref{table2} with $\mu_2 =0.65$, $\mu_3=0.9$, $p_2=0.001$, $p_3=0.05$, $k_1=500$, $k_2=5$, $\eta=0.05$, $\tau=7$ and $10^{-9} < \beta < 10^{-7}$,}
	\label{Fig:threshold_r0}
\end{figure}

\subsection{Model validation using real data}
SARS-CoV-2 viral load data are obtained from Wolfel et al. \cite{wolfel2020virological}.  
They studied patients from a hospital in Munich, Germany. They reported Daily measurements of viral load in sputum, pharyngeal swabs and stool for 9 patients. Among these patients, there were two patients (namely, patient A and patient B) for whom the growth phase of sputum data was captured. We therefore utilized these two datasets for our analysis. The data was collected from Wolfel et al. \cite{wolfel2020virological} using a online software \cite{webplot2020}. 

The solution curve of viral load ($V(t)$) is fitted to data using the built-in (MATLAB, R2018a) simplex algorithm to minimize the sum of squares difference between simulated indicators and data. We used the MATLAB function `fminsearchbnd' to perform the optimization. During the computation, 100 different starting points in parameter space were chosen using Latin Hypercube Sampling to ensure consistency and uniqueness of the parameter estimates. The fitting is displayed in Fig. \ref{Fig:model_fitting}(a) for patient A and in Fig. \ref{Fig:model_fitting}(b) for patient B. The fixed parameters are taken from Table \ref{table2} with $\mu_2 =0.65$, $k_2=5$ and $\eta=0.05$. The initial conditions are taken as mentioned in the beginning of Section \ref{numerical}. We estimated five parameters directly related to viral load of a patient viz., $\beta$, $k_1$, $p_2$, $p_3$ and $\mu_3$. The estimated parameters for patient A are found to be $\beta=1.7505 \times 10^{-6}$, $k_1=379$, $p_2=0.2805$, $p_3=0.0316$ and $\mu_3=0.8108$. Similarly, the estimates for patient B are obtained as $\beta=5.561 \times 10^{-7}$, $k_1=128$, $p_2=0.9403$, $p_3=0.0057$ and $\mu_3=0.99$.

\begin{figure}[t]
	\includegraphics[width=0.45\textwidth]{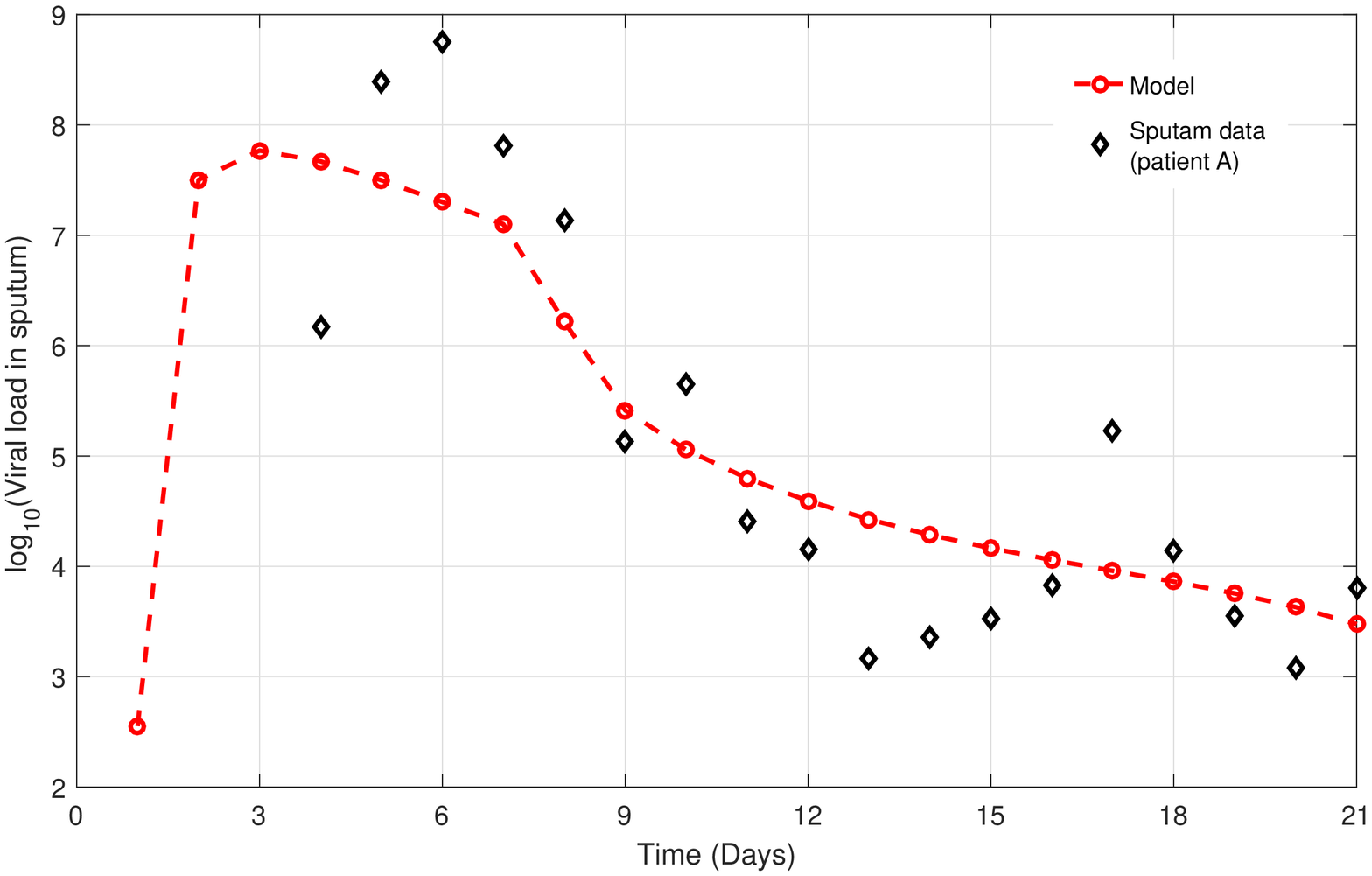}{a}
	\includegraphics[width=0.45\textwidth]{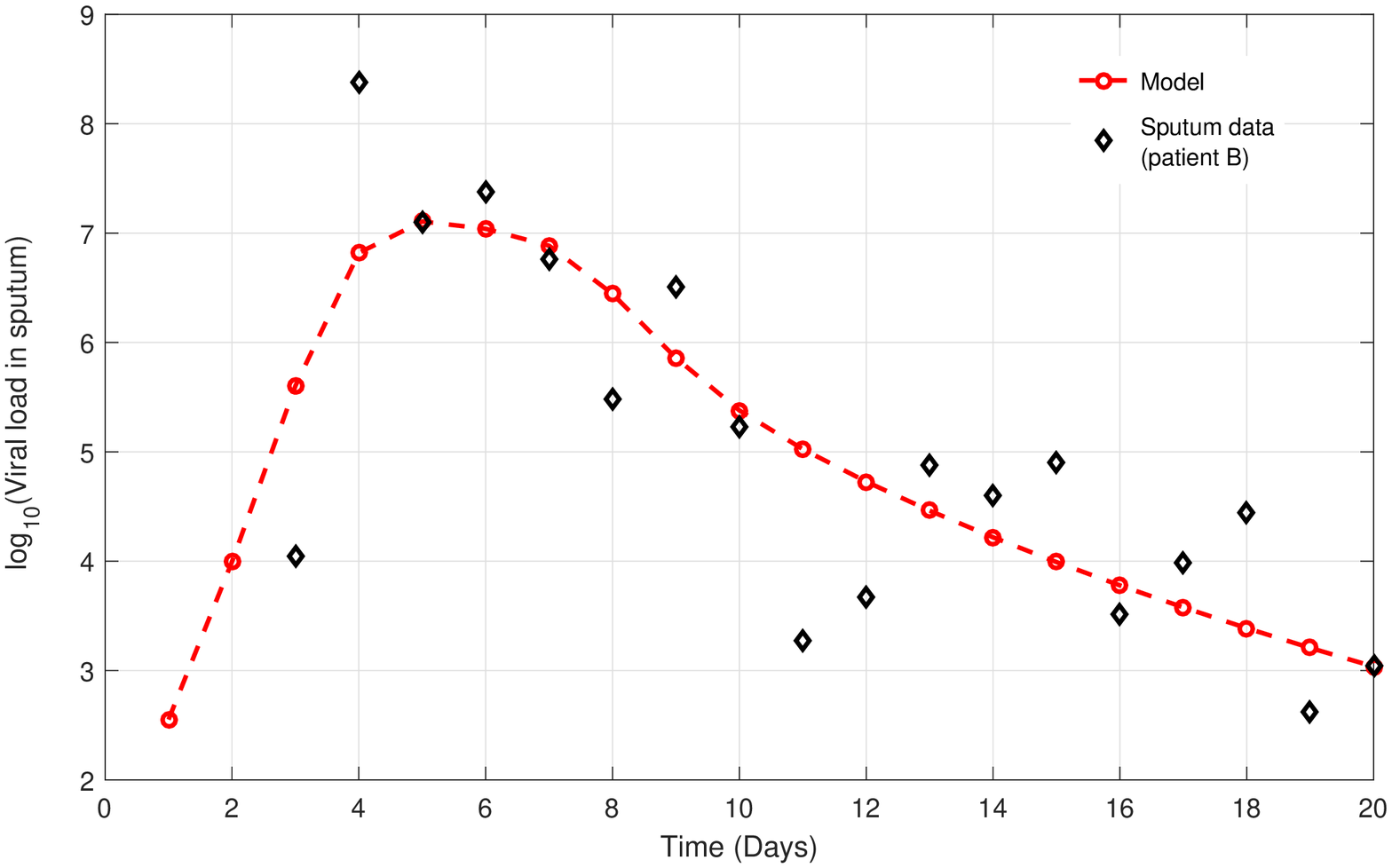}{b}
	\caption{Fitting model solution to (a) patient A data and (b) patient B data.}
	\label{Fig:model_fitting}
\end{figure}

\subsection{Sensitivity analysis}
We performed global sensitivity analysis to identify most influential parameters with respect to the maximum size (or alternatively, the peak of load) of virus particles ($V_{max}$) in 3 months time frame. Partial rank correlation coefficients (PRCCs) are calculated and plotted in Fig. \ref{Fig:PRCCs}. Nonlinear and monotone relationship were observed for the parameters with respect to $V_{max}$, which is a prerequisite for performing PRCC analysis. Following Marino et. al \cite{marino2008methodology}, we calculate PRCCs for the parameters $\beta$, $k_1$, $k_2$, $\mu_2$, $\mu_3$, $p_2$, $p_3$, $\gamma$ and $\eta$. The base values for the parameters $\beta$, $k_1$, $p_2$, $p_3$ and $\mu_3$ are taken as the average of estimated parameters of patient A and patient B. The other base values are $\mu_2 =0.65$, $k_2=5$, $\gamma=0.5$ and $\eta=0.05$. For each of the parameters, 500 Latin Hypercube Samples were generated from the interval (0.5 $\times$ base value, 1.5 $\times$ base value).

\begin{figure}[t]
	\includegraphics[width=1.0\textwidth]{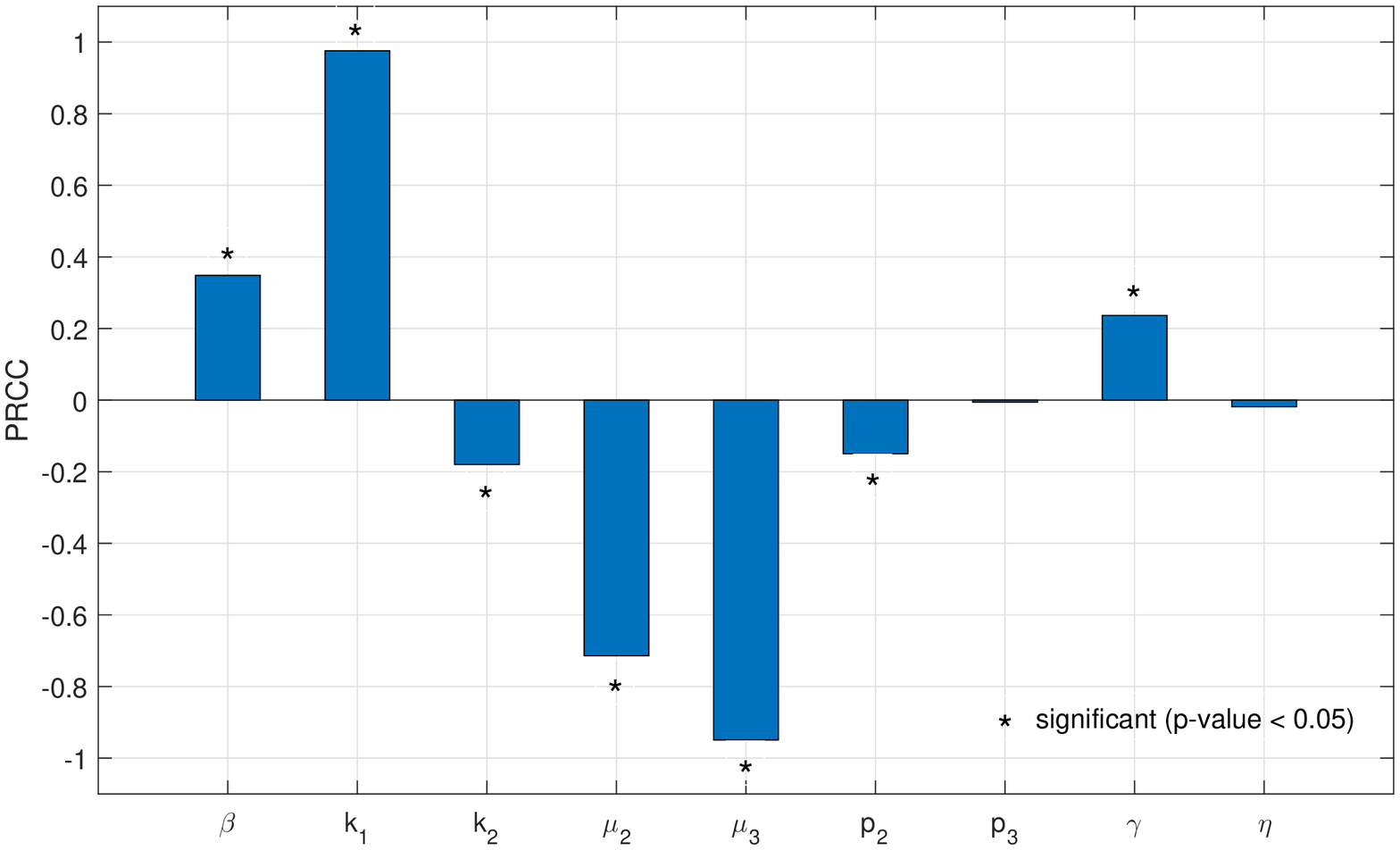}
	\caption{Effect of uncertainty of the model (\ref{eq1}) on the peak of viral load. Parameters with significant PRCC indicated as $^*$ (p-value $<$ 0.05). The fixed parameters are taken from Table \ref{table2} with $\mu_2 =0.65$, $k_2=5$, $\eta=0.05$ and $\tau=7$.}
	\label{Fig:PRCCs}
\end{figure}

It is observed that the parameters $\beta$, $k_1$ and $\gamma$ has significant positive correlations with $V_{max}$. This indicates that the production rate of virus particles from infected cells will increase the chance of larger infection propagation. Besides, the infection rate and the immiunosuppresion rate are positively correlated with the peak of viral load. On the other hand, the natural death rate of infected cells and death rate of virus particles will have significant negative correlation with $V_{max}$. The production rate of cytokines is also negatively correlated with $V_{max}$. These results reinforces the fact that $\beta$ and $k_1$ are very crucial for reduction of viral load. 

\section{Model with antiviral treatment}\label{treatment}
Antiviral drugs can be used to slow SARS-CoV-2 infection or block production of virus particles. These drugs will necessarily save the lives of many severely ill patients and will reduce the time spent in intensive care units for patients, vacating hospital beds. Antiviral medications will, in turn, inhibit subsequent transmission that could happen if the drugs were not given. However, to analyze the effect of antiviral treatment, we consider drugs can block infection and/or production of virus particles. Many studies have suggested various existing compounds for testing \cite{tay2020trinity,encinar2020potential,caly2020fda} as SARS-CoV-2 antiviral drug, but World Health Organization (WHO) is focusing on the following four therapies: an experimental antiviral compound called remdesivir; the malaria medications chloroquine and hydroxychloroquine; a combination of two HIV drugs, lopinavir and ritonavir; and that same combination plus interferon-beta, an immune system messenger that can help cripple viruses \cite{kupferschmidt2020launches}.

Following Zitzmann et al. \cite{zitzmann2018mathematical}, we incorporate antiviral drug treatment in the proposed model \eqref{eq1}. The modified system with antiviral treatment is given by

\begin{eqnarray}\label{eq1_treatment}
\frac{dH}{dt}&=& \Pi- (1-\epsilon_1)\beta HV -\mu_1 H,\nonumber\\
\frac{dI}{dt}&=& (1-\epsilon_1)\beta HV- p_1 TI - p_5 NI - \mu_2 I,\nonumber\\
\frac{dV}{dt}&=& (1-\epsilon_2) k_1 I - p_2 CV - p_3 AV - \mu_3 V,\nonumber\\
\frac{dC}{dt}&=& \frac{k_2 I}{1 + \gamma V} - \mu_4 C,\\
\frac{dN}{dt}&=& rC - \mu_5 N,\\
\frac{dT}{dt}&=& \lambda_1 TC - \mu_6 T,\nonumber\\
\frac{dB}{dt}&=& \lambda_2 BT - \mu_7 B,\nonumber\\
\frac{dA}{dt}&=& G(t-\tau)\eta B - p_4 AV -\mu_8 A.\nonumber 
\end{eqnarray}

\begin{figure}[t]
	\includegraphics[width=0.45\textwidth]{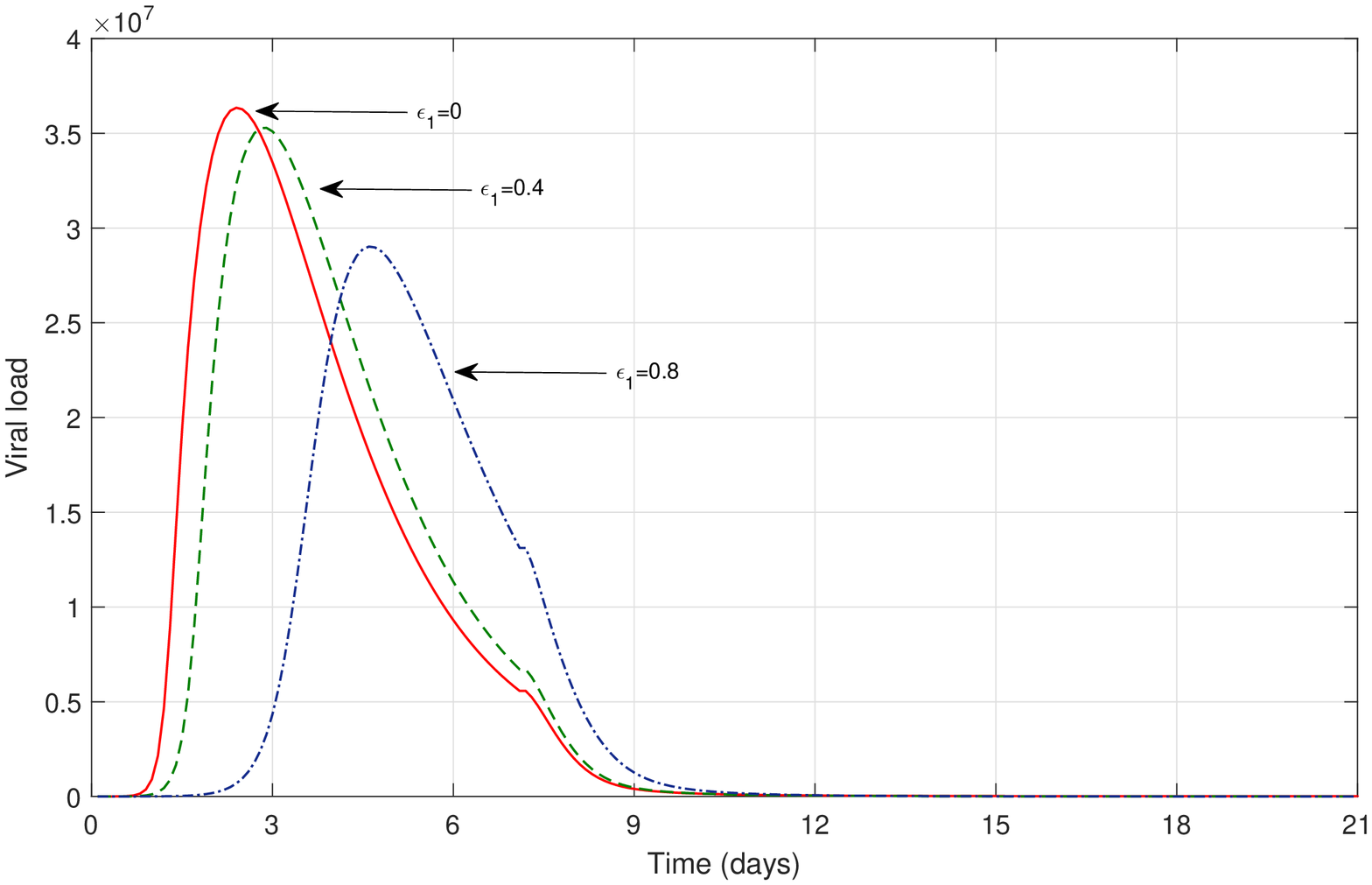}{a}
	\includegraphics[width=0.45\textwidth]{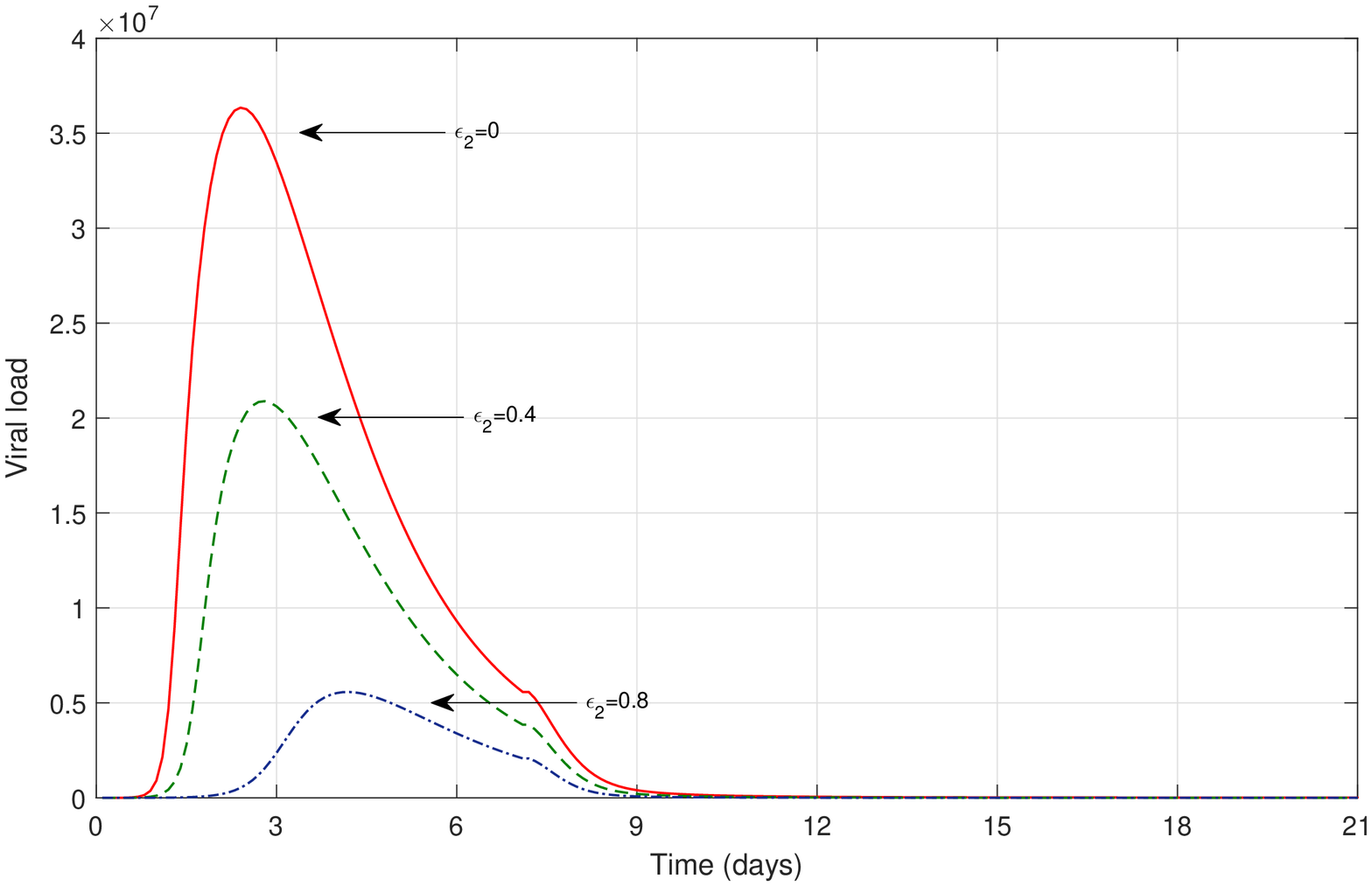}{b}
	\caption{Effect of antiviral drugs that (a) reduce infection or (b) blocks virus production. The time series of viral load is presented for different values of $\epsilon_1$ and $\epsilon_2$. The fixed parameters are taken from Table \ref{table2} with $\mu_2 =0.65$, $k_2=5$, $\eta=0.05$ and $\tau=7$. Other fixed values are taken to be the average of estimated parameters for patient A and patient B.}
	\label{Fig:treatment}
\end{figure}

From Fig. \ref{Fig:treatment}, it can be noted that increase in $\epsilon_1$ reduces the peak of viral load but the duration of high viral load remains same. On the other hand, increase in $\epsilon_2$ significantly reduce both peak of viral load and duration of high viral load. Thus, we conclude that blocking the virus production from infected cells is a more suitable target for antiviral drug development.

Finally, we study the effect vaccination in the viral dynamics of SARS-CoV-2 in humans. A vaccine is a biological preparation that provides active acquired immunity to a particular infectious agent. Thus if an individual is vaccinated, there will be no delay in the development of antibody. Therefore, the delay term $\tau$ is taken to be zero for vaccinated individuals (see Fig. \ref{Fig:treatment_vaccine}). It is observed that vaccination not only reduces the viral load in healthy patients but also reduces the duration of high viremia. 

\begin{figure}[t]
	\includegraphics[width=1.0\textwidth]{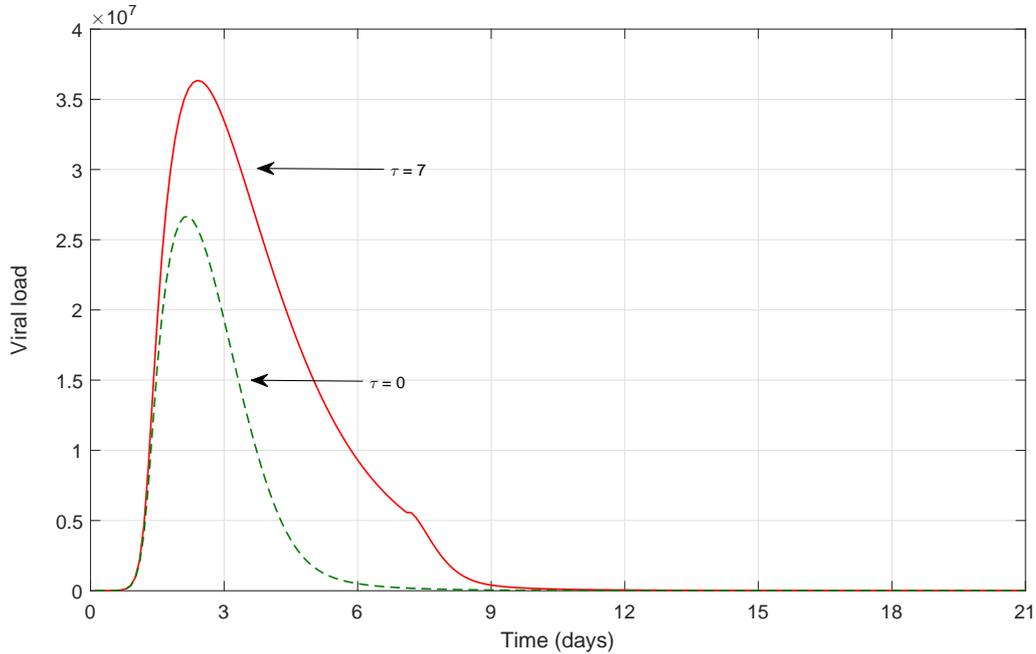}
	\caption{Viral load time series for different values of $\tau$ for the model \eqref{eq1}. The fixed parameters are taken from Table \ref{table2} with $\mu_2 =0.65$, $k_2=5$, $\eta=0.05$ and $\tau=7$. Other fixed values are taken to be the average of estimated parameters for patient A and patient B.}
	\label{Fig:treatment_vaccine}
\end{figure}

Overall, for antiviral drug target, blocking virus production is more fruitful in terms of viral load reduction and vaccination will also be effective.

\section{Discussion and conclusion}\label{discussion}
In this study, we have proposed and analyzed a compartmental model of SARS-CoV-2 transmission within the human body. The much needed innate and adaptive immune responses are incorporated into the model. The eight-dimensional model has four types of equilibrium points. The existence criterion for each type of equilibria is presented. From the local stability of the DFE, the expression for basic reproduction number is obtained. This number is very crucial for the persistence of the virus in the long run. However, the short-term dynamics of the viral load is studied using various numerical techniques. During time series analysis, we observed that the viral load time series experiences a peak between sixth and seven days post-infection, followed by a sharp decrease due to activation of adaptive immune response (see Fig. \ref{Fig:time_series_VA}). A forward bifurcation of equilibria with respect to the basic reproduction number is observed and depicted in Fig. \ref{Fig:threshold_r0}. This also ensures that if we suitably vary parameters involved in the expression of $R_0$, the same type of phenomenon occurs. Thus, in turn, parameters such as $\beta$ and $k_1$ can be decreased to reduce $R_0$ below unity and ensure local asymptotic stability of DFE. 

We used daily measurements of SARS-CoV-2 viral load in sputum for two patients \cite{wolfel2020virological} from a hospital in Munich, Germany. Using the estimated parameters, the global sensitivity analysis of several model parameters with respect to peak viral load is performed. The results indicate that the production rate of virus particles from infected cells will increase the chance of more significant infection propagation. Besides, the infection rate and the immiunosuppresion rate will increase the peak of viral load. Additionally, the natural death rates of infected cells and the death rate of virus particles will have a significant negative correlation with the peak of viral load. The production rate of cytokines is also negatively correlated with the peak of viral load. These results reinforce the fact that $\beta$ and $k_1$ are very crucial for the reduction of viral load. 

Antiviral drugs can be used to slow SARS-CoV-2 infection (or reduce $\beta$) or block the production of virus particles (or reduce $k_1$). Results suggest that a decrease in $\beta$ reduces the peak of viral load but the duration of the high viral load remains the same. On the other hand, a decrease in $k_1$ significantly reduce both peak of viral load and period of high viral load. Thus, we conclude that blocking virus production from infected cells is a more suitable target for antiviral drug development. Moreover, vaccination can reduce the viral load in healthy patients and also reduce the duration of high viremia in the body. But vaccine development is a complicated task; therefore, during the vaccine development phase, blocking virus production from infected cells can be targeted for antiviral drug development.  

Researchers have been putting more effort to develop a vaccine to tackle COVID-19 \cite{dhama2020covid,lurie2020developing}. The journey has started with the first clinical trial just two months after the genetic sequence of the virus. The mathematical model developed in this paper can be improved by adding more detailed data to reveal prophylactic and therapeutic interventions. Our theoretical findings should be tested clinically for the implementation. Further insights into immunology and pathogenesis of SARS-CoV-2 will help to improve the outcome of this and future pandemics.

\bibliographystyle{unsrt}
\bibliography{covid_within}

\begin{thebibliography}{10}

\bibitem{huang2020clinical}
Chaolin Huang, Yeming Wang, Xingwang Li, Lili Ren, Jianping Zhao, Yi~Hu,
  Li~Zhang, Guohui Fan, Jiuyang Xu, Xiaoying Gu, et~al.
\newblock Clinical features of patients infected with 2019 novel coronavirus in
  wuhan, china.
\newblock {\em The Lancet}, 395(10223):497--506, 2020.

\bibitem{gumel2004modelling}
Abba~B Gumel, Shigui Ruan, Troy Day, James Watmough, Fred Brauer, P~Van~den
  Driessche, Dave Gabrielson, Chris Bowman, Murray~E Alexander, Sten Ardal,
  et~al.
\newblock Modelling strategies for controlling sars outbreaks.
\newblock {\em Proceedings of the Royal Society of London. Series B: Biological
  Sciences}, 271(1554):2223--2232, 2004.

\bibitem{li2003angiotensin}
Wenhui Li, Michael~J Moore, Natalya Vasilieva, Jianhua Sui, Swee~Kee Wong,
  Michael~A Berne, Mohan Somasundaran, John~L Sullivan, Katherine Luzuriaga,
  Thomas~C Greenough, et~al.
\newblock Angiotensin-converting enzyme 2 is a functional receptor for the sars
  coronavirus.
\newblock {\em Nature}, 426(6965):450--454, 2003.

\bibitem{de2013commentary}
Raoul~J de~Groot, Susan~C Baker, Ralph~S Baric, Caroline~S Brown, Christian
  Drosten, Luis Enjuanes, Ron~AM Fouchier, Monica Galiano, Alexander~E
  Gorbalenya, Ziad~A Memish, et~al.
\newblock Commentary: Middle east respiratory syndrome coronavirus (mers-cov):
  announcement of the coronavirus study group.
\newblock {\em Journal of virology}, 87(14):7790--7792, 2013.

\bibitem{de2016sars}
Emmie de~Wit, Neeltje van Doremalen, Darryl Falzarano, and Vincent~J Munster.
\newblock Sars and mers: recent insights into emerging coronaviruses.
\newblock {\em Nature Reviews Microbiology}, 14(8):523, 2016.

\bibitem{cowling2015preliminary}
Benjamin~J Cowling, Minah Park, Vicky~J Fang, Peng Wu, Gabriel~M Leung, and
  Joseph~T Wu.
\newblock Preliminary epidemiologic assessment of mers-cov outbreak in south
  korea, may--june 2015.
\newblock {\em Euro surveillance: bulletin Europeen sur les maladies
  transmissibles= European communicable disease bulletin}, 20(25), 2015.

\bibitem{kim2017middle}
KH~Kim, TE~Tandi, Jae~Wook Choi, JM~Moon, and MS~Kim.
\newblock Middle east respiratory syndrome coronavirus (mers-cov) outbreak in
  south korea, 2015: epidemiology, characteristics and public health
  implications.
\newblock {\em Journal of Hospital Infection}, 95(2):207--213, 2017.

\bibitem{sardar2020realistic}
Tridip Sardar, Indrajit Ghosh, Xavier Rod{\'o}, and Joydev Chattopadhyay.
\newblock A realistic two-strain model for mers-cov infection uncovers the high
  risk for epidemic propagation.
\newblock {\em PLoS neglected tropical diseases}, 14(2):e0008065, 2020.

\bibitem{kwok2019epidemic}
Kin~On Kwok, Arthur Tang, Vivian~WI Wei, Woo~Hyun Park, Eng~Kiong Yeoh, and
  Steven Riley.
\newblock Epidemic models of contact tracing: Systematic review of transmission
  studies of severe acute respiratory syndrome and middle east respiratory
  syndrome.
\newblock {\em Computational and structural biotechnology journal}, 2019.

\bibitem{dhama2020covid}
Kuldeep Dhama, Khan Sharun, Ruchi Tiwari, Maryam Dadar, Yashpal~Singh Malik,
  Karam~Pal Singh, and Wanpen Chaicumpa.
\newblock Covid-19, an emerging coronavirus infection: advances and prospects
  in designing and developing vaccines, immunotherapeutics, and therapeutics.
\newblock {\em Human Vaccines \& Immunotherapeutics}, pages 1--7, 2020.

\bibitem{lurie2020developing}
Nicole Lurie, Melanie Saville, Richard Hatchett, and Jane Halton.
\newblock Developing covid-19 vaccines at pandemic speed.
\newblock {\em New England Journal of Medicine}, 2020.

\bibitem{menachery2015sars}
Vineet~D Menachery, Boyd~L Yount~Jr, Kari Debbink, Sudhakar Agnihothram, Lisa~E
  Gralinski, Jessica~A Plante, Rachel~L Graham, Trevor Scobey, Xing-Yi Ge,
  Eric~F Donaldson, et~al.
\newblock A sars-like cluster of circulating bat coronaviruses shows potential
  for human emergence.
\newblock {\em Nature medicine}, 21(12):1508, 2015.

\bibitem{vargas2020host}
Esteban Abelardo~Hernandez Vargas and Jorge~X Velasco-Hernandez.
\newblock In-host modelling of covid-19 kinetics in humans.
\newblock {\em medRxiv}, 2020.

\bibitem{hu2020insights}
Tony~Y Hu, Matthew Frieman, and Joy Wolfram.
\newblock Insights from nanomedicine into chloroquine efficacy against
  covid-19.
\newblock {\em Nature Nanotechnology}, 15(4):247--249, 2020.

\bibitem{yaqinuddin2020innate}
Ahmed Yaqinuddin and Junaid Kashir.
\newblock Innate immunity in covid-19 patients mediated by nkg2a receptors, and
  potential treatment using monalizumab, cholroquine, and antiviral agents.
\newblock {\em Medical Hypotheses}, page 109777, 2020.

\bibitem{tay2020trinity}
Matthew~Zirui Tay, Chek~Meng Poh, Laurent R{\'e}nia, Paul~A MacAry, and Lisa~FP
  Ng.
\newblock The trinity of covid-19: immunity, inflammation and intervention.
\newblock {\em Nature Reviews Immunology}, pages 1--12, 2020.

\bibitem{wu2020nowcasting}
Joseph~T Wu, Kathy Leung, and Gabriel~M Leung.
\newblock Nowcasting and forecasting the potential domestic and international
  spread of the 2019-ncov outbreak originating in wuhan, china: a modelling
  study.
\newblock {\em The Lancet}, 395(10225):689--697, 2020.

\bibitem{tang2020updated}
Biao Tang, Nicola~Luigi Bragazzi, Qian Li, Sanyi Tang, Yanni Xiao, and Jianhong
  Wu.
\newblock An updated estimation of the risk of transmission of the novel
  coronavirus (2019-ncov).
\newblock {\em Infectious Disease Modelling}, 2020.

\bibitem{kucharski2020early}
Adam~J Kucharski, Timothy~W Russell, Charlie Diamond, Yang Liu, John Edmunds,
  Sebastian Funk, Rosalind~M Eggo, Fiona Sun, Mark Jit, James~D Munday, et~al.
\newblock Early dynamics of transmission and control of covid-19: a
  mathematical modelling study.
\newblock {\em The lancet infectious diseases}, 2020.

\bibitem{kochanczyk2020dynamics}
Marek Kocha{\'n}czyk, Frederic Grabowski, and Tomasz Lipniacki.
\newblock Dynamics of covid-19 pandemic at constant and time-dependent contact
  rates.
\newblock {\em Mathematical Modelling of Natural Phenomena}, 15:28, 2020.

\bibitem{du2020mathematical}
Sean~Quan Du and Weiming Yuan.
\newblock Mathematical modeling of interaction between innate and adaptive
  immune responses in covid-19 and implications for viral pathogenesis.
\newblock {\em Journal of Medical Virology}, 2020.

\bibitem{tufan2020covid}
Abdurrahman Tufan, ASLIHAN~AVANO{\u{G}}LU G{\"U}LER, and Marco Matucci-Cerinic.
\newblock Covid-19, immune system response, hyperinflammation and repurposing
  antirheumatic drugs.
\newblock {\em Turkish Journal of Medical Sciences}, 50(SI-1):620--632, 2020.

\bibitem{mckechnie2020innate}
Julia~L McKechnie and Catherine~A Blish.
\newblock The innate immune system: fighting on the front lines or fanning the
  flames of covid-19?
\newblock {\em Cell Host \& Microbe}, 2020.

\bibitem{ciupe2017host}
Stanca~M Ciupe and Jane~M Heffernan.
\newblock In-host modeling.
\newblock {\em Infectious Disease Modelling}, 2(2):188--202, 2017.

\bibitem{sasmal2017mathematical}
Sourav~Kumar Sasmal, Yueping Dong, and Yasuhiro Takeuchi.
\newblock Mathematical modeling on t-cell mediated adaptive immunity in primary
  dengue infections.
\newblock {\em Journal of theoretical biology}, 429:229--240, 2017.

\bibitem{canini2011population}
Laetitia Canini and Fabrice Carrat.
\newblock Population modeling of influenza a/h1n1 virus kinetics and symptom
  dynamics.
\newblock {\em Journal of virology}, 85(6):2764--2770, 2011.

\bibitem{ben2015minimal}
Rotem Ben-Shachar and Katia Koelle.
\newblock Minimal within-host dengue models highlight the specific roles of the
  immune response in primary and secondary dengue infections.
\newblock {\em Journal of the Royal Society Interface}, 12(103):20140886, 2015.

\bibitem{raoult2020coronavirus}
Didier Raoult, Alimuddin Zumla, Franco Locatelli, Giuseppe Ippolito, and Guido
  Kroemer.
\newblock Coronavirus infections: Epidemiological, clinical and immunological
  features and hypotheses.
\newblock {\em Cell Stress}, 4(4):66, 2020.

\bibitem{fowler1981approximate}
AC~Fowler.
\newblock Approximate solution of a model of biological immune responses
  incorporating delay.
\newblock {\em Journal of mathematical biology}, 13(1):23--45, 1981.

\bibitem{gujarati2014virus}
Tanvi~P Gujarati and G~Ambika.
\newblock Virus antibody dynamics in primary and secondary dengue infections.
\newblock {\em Journal of mathematical biology}, 69(6-7):1773--1800, 2014.

\bibitem{nikin2015role}
Ryan Nikin-Beers and Stanca~M Ciupe.
\newblock The role of antibody in enhancing dengue virus infection.
\newblock {\em Mathematical biosciences}, 263:83--92, 2015.

\bibitem{clapham2014within}
Hannah~E Clapham, Vianney Tricou, Nguyen Van Vinh~Chau, Cameron~P Simmons, and
  Neil~M Ferguson.
\newblock Within-host viral dynamics of dengue serotype 1 infection.
\newblock {\em Journal of the Royal Society Interface}, 11(96):20140094, 2014.

\bibitem{immunity2020who}
WHO.
\newblock "immunity passports" in the context of covid-19.
\newblock
  \url{https://www.who.int/news-room/commentaries/detail/immunity-passports-in-the-context-of-covid-19}.
\newblock Retrieved : 2020-05-26.

\bibitem{wolfel2020virological}
Roman W{\"o}lfel, Victor~M Corman, Wolfgang Guggemos, Michael Seilmaier, Sabine
  Zange, Marcel~A M{\"u}ller, Daniela Niemeyer, Terry~C Jones, Patrick Vollmar,
  Camilla Rothe, et~al.
\newblock Virological assessment of hospitalized patients with covid-2019.
\newblock {\em Nature}, pages 1--5, 2020.

\bibitem{webplot2020}
WPD.
\newblock Online software for data extraction.
\newblock \url{https://apps.automeris.io/wpd/}.
\newblock Retrieved : 2020-05-20.

\bibitem{marino2008methodology}
Simeone Marino, Ian~B Hogue, Christian~J Ray, and Denise~E Kirschner.
\newblock A methodology for performing global uncertainty and sensitivity
  analysis in systems biology.
\newblock {\em Journal of theoretical biology}, 254(1):178--196, 2008.

\bibitem{encinar2020potential}
Jos{\'e}~Antonio Encinar and Javier~A Menendez.
\newblock Potential drugs targeting early innate immune evasion of
  sars-coronavirus 2 via 2’-o-methylation of viral rna.
\newblock {\em Viruses}, 12(5):525, 2020.

\bibitem{caly2020fda}
Leon Caly, Julian~D Druce, Mike~G Catton, David~A Jans, and Kylie~M Wagstaff.
\newblock The fda-approved drug ivermectin inhibits the replication of
  sars-cov-2 in vitro.
\newblock {\em Antiviral research}, page 104787, 2020.

\bibitem{kupferschmidt2020launches}
Kai Kupferschmidt and Jon Cohen.
\newblock Who launches global megatrial of the four most promising coronavirus
  treatments.
\newblock {\em Science}, 22, 2020.

\bibitem{zitzmann2018mathematical}
Carolin Zitzmann and Lars Kaderali.
\newblock Mathematical analysis of viral replication dynamics and antiviral
  treatment strategies: from basic models to age-based multi-scale modeling.
\newblock {\em Frontiers in microbiology}, 9:1546, 2018.

\end{thebibliography}
\end{document}